\documentclass[conference]{IEEEtran}

\usepackage{amssymb}
\usepackage{amsthm}
\usepackage{algorithm}
\usepackage{algpseudocode}
\usepackage{graphicx}
\usepackage[table,xcdraw]{xcolor}

\newtheorem{definition}{Definition}
\newtheorem{theorem}{Theorem}
\newtheorem{lemma}{Lemma}[theorem]

\begin{document}


\title{PBL: System for Creating and Maintaining Personal Blockchain Ledgers}

\author{\IEEEauthorblockN{Collin Connors}
\IEEEauthorblockA{\textit{Computer Science Department} \\
\textit{University of Miami}\\
Miami, FL \\
cdc104@miami.edu}
\and
\IEEEauthorblockN{Dilip Sarkar}
\IEEEauthorblockA{\textit{Computer Science Department} \\
\textit{University of Miami}\\
Miami, FL}
}

\maketitle

\begin{abstract} 


Blockchain technology has experienced substantial growth in recent years, yet the diversity of blockchain applications has been limited. Blockchain provides many desirable features for applications, including being append-only, immutable, tamper-evident, tamper-resistant, and fault-tolerant; however, many applications that would benefit from these features cannot incorporate current blockchains. This work presents a novel architecture for creating and maintaining personal blockchain ledgers that address these concerns. Our system utilizes independent modular services, enabling individuals to securely store their data in a personal blockchain ledger. Unlike traditional blockchain, which stores all transactions of multiple users, our novel personal blockchains are designed to allow individuals to maintain their privacy without requiring extensive technical expertise. Using rigorous mathematical methods, we prove that our system produces append-only, immutable, tamper-evident, tamper-resistant ledgers. Our system addresses use cases not addressed by traditional blockchain development platforms. Our system creates a new blockchain paradigm, enabling more individuals and applications to leverage blockchain technology for their needs.

\end{abstract}

\section{Introduction}
\label{sec:Introduction}

While blockchain technology has been touted as the future of data storage, its breadth is currently limited by the available blockchain development platforms. Blockchain first saw use as a means for storing Bitcoin transactions~\cite{NakamotoBitcoin}. As the technology has become more popular, other more extensive blockchains development platforms have arisen, such as Ethereum~\cite{EthereumWhitepaper} and Hyperledger Fabric~\cite{HyperledgerFabricWhitepaper}. While second-generation blockchains allow users to store more complex data than their first-generation counterparts, they are not ideal for all use cases where a blockchain is desired.

Storing data on a blockchain offers many unique befits. One property of blockchains is that a blockchain is an append-only ledger; thus, data in a blockchain is always ordered. The first piece of data will always be at the start of the blockchain, and the most recently added data will always be at the end. Likewise, blockchains are immutable. Once an order is imposed on the blockchain, it cannot be changed. Blockchains are tamper-evident and tamper-resistant. The structure of a blockchain makes it evident if any changes were made. Similarly, it would require substantial work to change the entire blockchain. Lastly, blockchains are fault-tolerant. Even if part of the blockchain fails, a fault occurs, the blockchain is still available. These five properties, append-only, immutability, tamper-evident, tamper-resistant, and fault tolerance, make blockchain a desirable way to store data in many applications.  

However, blockchain has its challenges. Permissionless blockchains, like those used for cryptocurrencies, allow anyone to read the data stored in the blockchain. These blockchains are not ideal for applications that require storing sensitive data. In contrast, the permissioned blockchains used in enterprise applications provide this extra layer of privacy, but they require extensive technical expertise to manage and maintain. These types of blockchains are not suitable for individuals with little technical knowledge. 

In this work, we propose a third type of ledger, a personal blockchain. Our personal blockchain ledgers maintain all of the desirable features of traditional blockchains while combining the accessibility of a permissionless blockchain and the privacy of a permissioned blockchain. However, unlike traditional blockchains, our ledgers do not store an entire network's data but rather only an individual user's data. Thus we are not proposing a system to replace traditional blockchains but rather a new paradigm to satisfy the areas missed by current blockchains. This new paradigm of blockchain ledger will allow for more diverse applications that could not be achieved through traditional blockchains. 

This work proposes a new system for creating and maintaining personal blockchain ledgers. We provide some background on key terminology and concepts used and our motivations behind creating such a system. We first outline the structure of our ledgers and provide mathematical definitions for a valid ledger in our system. We then describe the necessary components required and how these components interact to create personal ledgers. We then prove that our ledgers hold all of the desirable features of a blockchain. This includes proving that our ledgers are append-only, immutable, tamper-evident, tamper-resistant, and fault-tolerant. Furthermore, we discuss the limitations of our system and how our ledgers can be applied in a practical use case to improve users' privacy and data security. We end with some concluding remarks on future work we have planned.  

\section{Background} 
\label{sec:Background}
Blockchains are  \emph{append-only}, \emph{immutable}, \emph{tamper-evident}, \emph{tamper-resistant}, \emph{distributed fault-tolerant} digital ledgers~\cite{NISTBlockchain}. As the name suggests, a blockchain is a linked chain of blocks. A block consists of the block header and the block data. 

The block header contains metadata about the block. Critically, the block header contains the cryptographic hash of the block data. The block data contains a list of transactions.

While initially, in Bitcoin, transactions were actual cryptocurrency transactions, modern blockchains allow transactions to be any arbitrary data. For example, if a doctor's office was using a blockchain to track appointments, the transactions might be appointment dates and who made the appointment. 

A blockchain is append-only. Thus, new blocks can only be added to the end of the blockchain. Blocks cannot be inserted between two existing blocks or prepended to the start of the blockchain. 

A blockchain is also immutable. The order of blocks cannot be changed in a blockchain. Likewise, the order of transactions cannot be changed in a block. This means that blockchains have a strict unchanging order of transactions and blocks.

A blockchain is tamper-evident and tamper-resistant. If a change is made to a single block, it is evident from the rest of the blockchain that it was made. Likewise, changes must be made to multiple blocks to make even a small change to a single block.

Lastly, a blockchain is fault tolerant. If some of the components fail, for example, by not responding to the rest of the system, the blockchain maintains availability. For example, the blockchain is still available if some of the Full Nodes or Miners go offline in Bitcoin. 

Blockchains rely on cryptographic hashes. A cryptographic hash function is a cryptographically secure function that maps some input data to a fixed length output \cite{cryptographicHashFnReview2012}. Cryptographic hash functions are deterministic, quick to compute, preimage resistant, second preimage resistant, and collision resistant. Some well-known cryptographic hash functions include SHA256, MD5, and RIPEMD160. 

One application of cryptographic hashes used in many blockchains is Merkle trees \cite{MerkleTrees}. A Merkle tree is a tree whose leaf nodes are data points, for example, transactions in a block. The parent of any two nodes is the cryptographic hash of the concatenated left and right children. Because cryptographic hash functions are collision resistant, any small change to the data will cause a change in the root of the Merkle tree. 

\subsection{Related Work}
In previous work, we analyzed the 23 most popular blockchain development platforms~\cite{BCCA_2022}. In this work, we detailed the various features of blockchains created by these platforms. During this process, we noticed that the current state of blockchain development platforms only covers some of the desired use cases. Thus in this work, we will propose a new paradigm for creating blockchains to supplement existing blockchain development. 

Our previous work identified the most popular permissioned blockchain as HyperLedger Fabric~\cite{HyperledgerFabricWhitepaper}. Fabric offers many unique and desirable features for creating blockchain ledgers. Specifically, Fabric offers a new model for creating ledgers, the Execute-Order-Validate (EOV) model. This new way of verifying transactions simplifies the process of creating applications. 

While Fabric provides an excellent solution to many of the problems posed by permissionless blockchains, it often does so at the expense of accessibility. We believe that a new blockchain paradigm that combines the ease of use from permissionless systems and the privacy from permissioned systems needs to exist.

Other systems try to simplify the process for the user by acting as a blockchain database. One such example is the open-source BigChainDB~\cite{BigChainDB}. This blockchain database allows users to store "assets" in a blockchain. BigChainDB assets can be digital representations of physical assets, such as a digital representation of a bike, or the asset can be a digital document, such as a digital healthcare record. Users can trade ownership of assets using a model similar to Bitcoin's UTXO model.  

While BigChainDB does make blockchain more accessible, it is limited by its asset-focused design. This model works well when the user wants to use blockchain for applications such as cryptocurrencies or NFTs; however, it must be optimized for other use cases. We believe that for blockchain to become more accessible, platforms need to shift their focus away from cryptocurrencies and focus on more general applications. 

Another popular blockchain database is Amazon's QLDB~\cite{Amazon_QLDB}. This database gives users a personal cryptographically secure digital ledger. We agree with QLDB's approach of creating individual ledgers; however, QLDB does not provide the necessary privacy for sensitive applications. Users are forced to use Amazon as their sole service provider. This lack of flexibility is not desirable in most applications. 

Like traditional centralized databases, QLDB gives a central entity, Amazon, potential access to a user's data. While many users may trust Amazon to respect their privacy, systems should be in place to ensure that no one entity has access to a user's data. Likewise, when using QLDB, it can be extremely difficult for a user to switch service providers. After deciding to use QLDB, an application can be stuck using Amazon as a service provider. Users should be free to choose their own service providers and change them without disrupting their ledger. 

\subsection{Motivation}
\label{sec:Motivation}
While blockchain technology can enhance many applications, the current blockchain architectures have limitations. Blockchain networks must choose between decentralization, scalability, and  privacy~\cite{Trilemma1997}. Thus no blockchain network is optimal in all three categories. Many blockchains chose to optimize decentralization at the cost of scalability and privacy. 

Blockchains fall into two main categories, permissionless blockchains, such as Ethereum, and permissioned blockchains, such as HyperLedger Fabric.

Permissionless blockchains have a cost when executing transactions. One such cost is the transaction fee users must pay to use the network. While some blockchains have attempted to ensure low transaction fees, the cost still turns many users off. Likewise, users may need to wait for transactions to be approved. Some blockchains have attempted to speed up transaction times, but typically, these solutions could be more scalable. 

Users who wish to avoid the drawbacks of permissionless blockchains may consider permissioned blockchains. However, these blockchains require technical expertise from the users to set up and maintain. Likewise, users need their own hardware to operate these blockchains, making them inaccessible to the average consumer.

We believe a \emph{blockchain system} ought to be a collection of \textbf{\emph{independent services }}. Different independent operators should provide these services. Thus, future users can deploy the blockchain system by choosing each service from available providers.

Likewise, users should be allowed to control their own data. Users should not have to make their data available to anyone they do not trust. Furthermore, no party other than the user should have full access to a user's data. 

We have proposed a novel system for creating personal blockchain ledgers to solve the problems with traditional blockchains. In the next section, we define a valid ledger in our system. Then we provide details on a system of independent services to create and maintain our ledgers. We then prove through mathematical methods that our proposed systems ledgers maintain the desired properties of blockchains. We then discuss how our system can improve applications. In section~\ref{sec:UseCases}, we provide an example use case that traditional blockchains fail to address and highlight how our system improves on this system. We are not proposing a system to replace traditional blockchains but rather a system to enhance blockchain by allowing this technology to be applied to a more diverse set of applications.

\section{Anatomy of a Ledger} 
\label{sec:AnatomyOfALedger}

Similar to traditional blockchains, our ledgers are made up of blocks connected via cryptographic hashes. Our model has two types of blocks: data blocks and genesis blocks. Each ledger has exactly one genesis block, with all other blocks being data blocks. In this section, we dissect the various components of our blocks to provide the background necessary to define our ledgers formally. 

\subsection{Data Blocks}
\label{sec:DataBlocks}
While our architecture relies on key ideas from traditional blockchains, our ledgers differentiate themselves by requiring signatures for all transactions. Figure~\ref{fig:DataBlock} shows a simplified anatomy of a data block in our ledger. Like traditional blockchains, our ledger separates a block into two portions, the block header containing metadata about the block and the block data containing the actual data stored in the block. 

\begin{figure}
	\centering
		\includegraphics[width=.75\columnwidth]{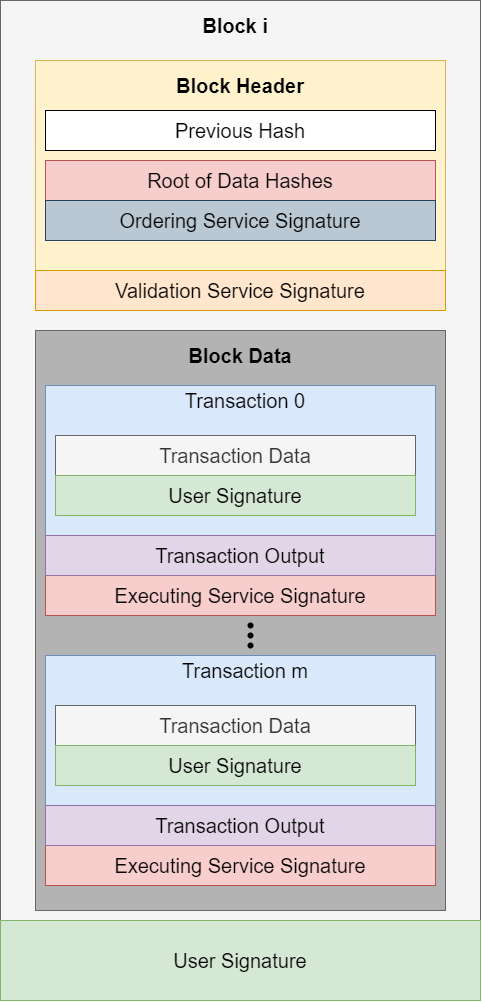}
	\caption{A simplified view of a data block in our personal blockchain system}
	\label{fig:DataBlock}
\end{figure}

The block data contains an ordered list of transactions. Each transaction is comprised of transaction data and the user's signature. The transaction data can be any data the user wishes to store on their ledger. Users often will not need to store the entire document on their ledger but only a document summary. For example, the user may store a monthly bank statement off-chain and store the hash of the statement on their ledger. Unlike other blockchains, this allows our system to be as lightweight as the user desires. The user must sign all transactions. This ensures that the user has acknowledged all transactions on their ledger and that only the user can submit transactions to their ledger. 

Similarly, a user may wish to have a third-party stakeholder sign the transaction to ensure high integrity. For example, suppose a user wants to share their monthly credit card statement ledger to be approved for a loan. In that case, the credit card company should sign each statement before the users submit the transaction. By having both the user's and the credit card company's signatures, the bank giving the loan can confirm that the data has not been falsified. 
As we will discuss in detail in section~\ref{sec:Architecture}, after a transaction is submitted to the ledger, it is evaluated by the executing service. The executing service performs any data-driven computations and appends the output to the transaction. For example, take a user that submits their credit card transactions to their ledger. The executing service can execute chaincode to calculate the user's balance after each transaction and append the output to the transaction. This process is similar to the execution of chaincode in Hyperledger Fabric~\cite{HyperledgerFabricWhitepaper}  

The executing service then signs the transaction and appends its signature to the transaction. This ensures that only the executing service can create transaction outputs. Since transaction outputs are generated from chaincode, the correctness of the outputs can easily be verified.

The block header contains information related to the block. While more information than shown in figure~\ref{fig:DataBlock} can be stored in the block header, such as block creation time, all block headers must contain the fields shown.
 
The previous hash is the hash of the previous block's header and the Validation Service signature. The previous hash cryptographically links each block to all previous blocks. 

Next, each block header contains the hash of all executing service signatures. This hash is generated using the Merkle Tree~\cite{MerkleTrees} of executing service signatures. The Merkle Tree ensures that all transactions are ordered. The ordering service is responsible for ordering the transactions and signing the result. This ensures that only the ordering service can order the transactions. 

Before adding a block to the ledger, the validation service validates the signatures in the block. After validating the signatures, it signs the block header. This ensures that all blocks submitted to the ledger will have valid signatures. 

Finally, to prove that the user, the only stakeholder in our personal ledger's architecture, has seen the block, the user signs it and appends it to their ledger. 

\subsection{Genesis Block}
\label{sec:GenesisBlock}

The first block in the ledger is a special block called the genesis block. Each ledger must start with a genesis block. Each ledger may have only one genesis block. Figure~\ref{fig:Genesis Block} shows a simplified example of a genesis block. 

\begin{figure}
	\centering
		\includegraphics[width=.75\columnwidth]{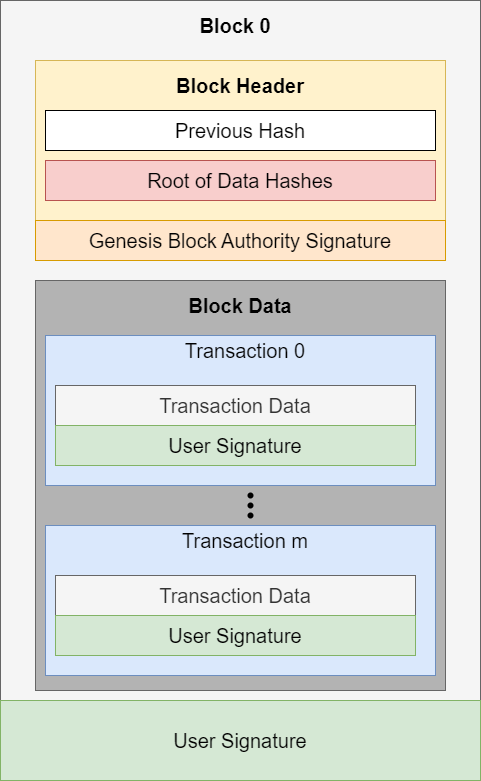}
	\caption{A simplified example of a Genesis Block in our personal blockchain architecture}
	\label{fig:Genesis Block}
\end{figure}

The genesis block contains transactions. However, these transactions contain data on the configuration of the ledger, unlike in a data block where transactions contain data the user wishes to store. These transactions may identify the user, the signing services, or other critical aspects of the ledger's configuration. A genesis block is not required to have any transactions if the user does not wish to provide additional configuration. These transactions have no output and are not processed by the executing service. 

Like data blocks, the genesis blocks have a previous hash field; however, in the genesis block, the previous hash field is set to all zeros since there is no previous block. In our ledger, the genesis block is the only block whose previous hash field is all zeros.

Likewise, the root of the data hashes is stored in the block header. If there are no transactions, this field is set to all zeros. The Genesis Block Authority confirms the validity of the ledger and signs the header. The Genesis Block Authority then sends the ledger back to the user.

Once the user signs the full genesis block, the block is considered valid and added as the first block in the user's ledger.

\subsection{Connections between Blocks}
\label{sec:ConnectionsBetweenBlocks}

Like in traditional blockchains, subsequent blocks in our ledger are connected via the previous hash. This field is the cryptographic hash of the previous block's header. As we will show shortly, this connection allows for many desirable properties. 

\begin{figure}
	\centering
		\includegraphics[width=\columnwidth]{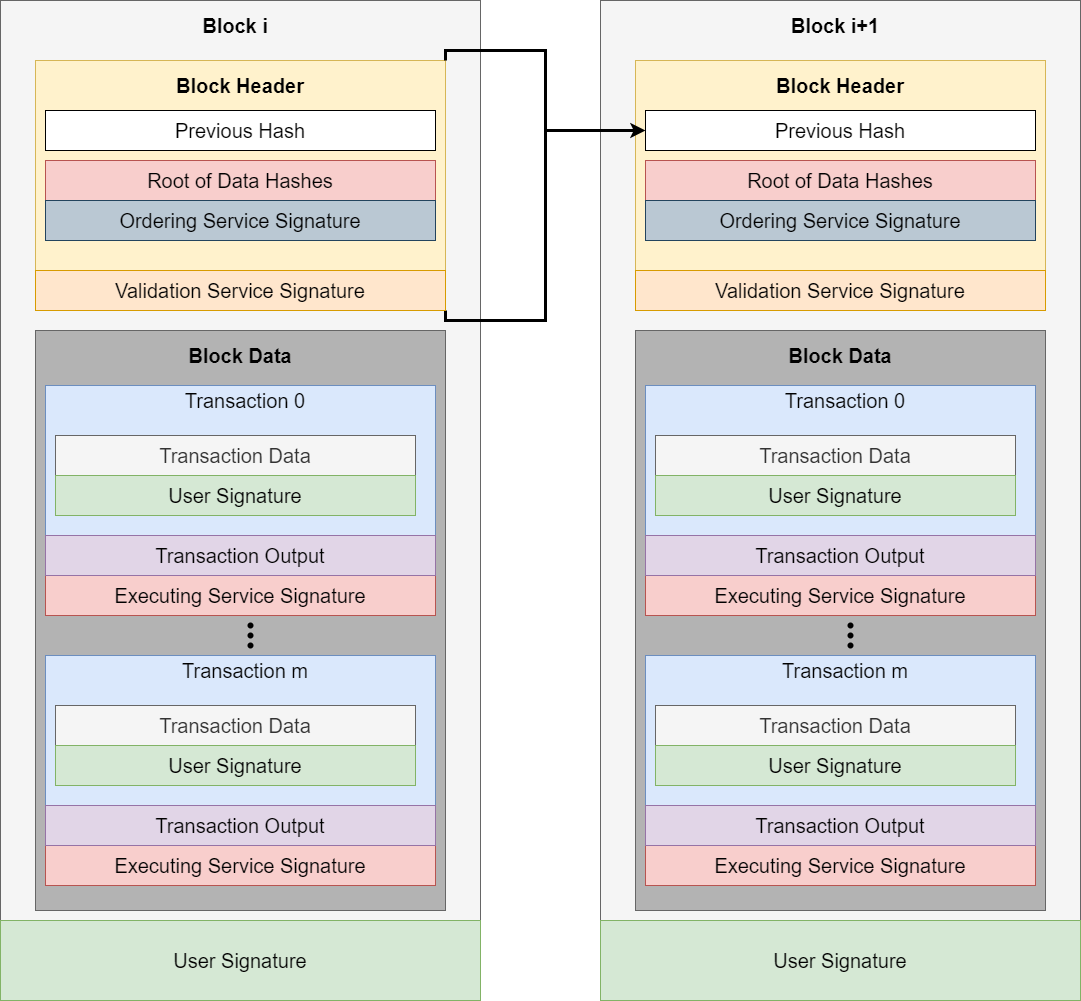}
	\caption{Visuliation of the connection between block i and its subsequent block}
	\label{fig:Connections.drawio}
\end{figure}

\subsection{Formal Definition our Digital Ledgers}
\label{sec:FormalDefinitionOfALedger}
Now that we have some background in the structure of our ledgers, we can formally define a valid ledger in our architecture. However, to define our ledgers properly, we must first define a valid genesis block, a valid data block, and a valid connection between two subsequent blocks.


\begin{definition}
A genesis block $B_0$ is valid if all the following statements are true:
\begin{enumerate}
	\item All required fields are present.
	\item The previous hash field of the genesis block is all 0's
	\item A Genesis Block Authority signs the hash of the genesis block header
	\item The User has signed the block
\end{enumerate}

\end{definition}

\begin{definition}
A data block $B$ is valid if all the following statements are true:
\begin{enumerate}
	\item All required fields are present.
	\item $B$.dataHash$=$Merkel($B.data$)
	\item The Executing Service signed all transactions
	\item The hash of the Merkel Tree of Executing Service signatures is in the header signed by the Ordering Service
	\item The hash of the block header is signed by the Validation service and appended to the header
	\item The User has signed the block
\end{enumerate}
\end{definition}

\begin{definition}
A connection $c(B_{i-1},B_i)$ is valid if all the following statement is true:
\begin{enumerate}
	\item $B_i$\texttt{.previousHash}$= h(B_{i-1}$\texttt{.header}$)$
\end{enumerate}
\end{definition}

\begin{definition}
A ledger $\mathcal{L}$ of length $n$ is valid if the following statements are true 
\begin{enumerate}
	\item The first block in the ledger is a valid genesis block
	\item There is only one genesis block; all other blocks are data blocks
	\item $\forall i \in [1,n]$ data block $B_i$ is valid
	\item $\forall i \in [1,n]$ each connection $c(B_{i-1},B_i)$ is valid
\end{enumerate}
\end{definition}

Using these definitions, we can formally prove that our ledgers maintain the desired properties of a blockchain. 

\section{Architecture} 
\label{sec:Architecture}
As briefly described in the previous section, our proposed architecture utilizes six independent modular services to create and maintain personal ledgers. The six services are:
\begin{enumerate}
\item The Ledger API - Used to facilitate requests between the users and the other services
\item The Storage Service - Responsible for interfacing with the user's desired storage location
\item The Genesis Block Authority - Responsible for creating the genesis block for each ledger
\item The Executing Service - Responsible for executing any data-driven code and signing incoming transactions
\item The Ordering Service - Responsible for forming blocks by ordering the user's transactions and signing the final order
\item The Validation Service - Responsible for ensuring the validity of a block before signing the block and sending it to the user for final approval 
\end{enumerate}

Our architecture is designed to create personal cryptographically linked digital ledgers. Unlike public blockchains such as Bitcoin, where the ledger must be distributed, our ledgers have no such restraint to maximize security and data privacy. However, the user may choose to store their ledgers in a distributed manner. We discuss the storage of the ledger further in this section. Likewise, notice that the block's integrity is no longer maintained by competing miners trying to satisfy a Sybil control mechanism rather through digital signatures.

\begin{figure}
	\centering
		\includegraphics[width=\columnwidth]{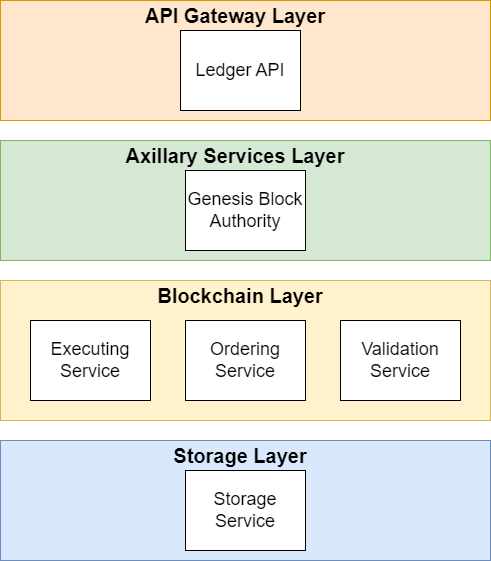}
	\caption{A bird's eye view of our proposed services}
	\label{fig:Services}
\end{figure}

The goal of this work is not to provide technical details on the implementation of each of these systems. Rather we briefly describe each service so that the reader understands the role each service plays in our system. More details on technical implementation for each service can be found on our GitHub page~\cite{CollinConnorsGitHubPage}.

\subsection{Ledger API}
\label{sec:LedgerAPI}
The first service that users interact with is the Ledger API. This service is an API interface that allows users to send requests to the remaining blockchain services. All of the users' interactions begin at the Ledger API. Our architecture gives users a single point of contact for the rest of the services to ensure ease of use. This prevents users from needing to understand complex relationships between services. 

The Ledger API is also responsible for generating the users' cryptographic keys and creating users' Root Addresses. Our model relies on users' ability to cryptographically sign transactions; thus, the users must have a public-private key pair. It is infeasible to expect users to keep track of their own key pairs; thus, our system utilizes a seed phrase that deterministically generates the key pair, similar to seed phrases described in BIP32~\cite{BIP32}. The seed phrase is a secure password allowing users to access their digital ledgers. Since seed phrases are generated from a list of 2048 words and must contain at least 12 words, it is infeasible to brute force a seed phrase. 

The Root Address is a blockchain address, similar to addresses found in Bitcoin, generated from a public-private key pair. This address is used to identify the user to all of the services. Public-private key pairs also allow users to verify that they own a ledger to any service provider. This is similar to how users can use their Bitcoin seed phrase to access their Bitcoin from any wallet provider. 

The Root Address generates each of the user's Ledger Addresses. Ledger addresses are used to identify specific ledgers owned by the user. For example, a user may have a ledger for storing financial transactions and a ledger for storing medical records. Each ledger has a unique Ledger Address generated from the user's Root Address. 

\subsection{Storage Service}
\label{sec:StorageService}
The storage service is responsible for interfacing with a user's storage location. The modular storage service allows users to select how the blockchains will get stored. One user may wish to have their storage service utilize IPFS, while another chooses to store the ledger in a centralized CouchDB database. Nevertheless, another user may store their ledger in cloud services like Google Drive. Users can store their data in multiple locations to increase data redundancy. We allow a broad spectrum of storage options to ensure our system meets users' needs. 

The storage service also stores metadata about the user's blockchains. The storage service stores the user's Root Address and any other Ledger Addresses the user generates. This allows a user to see all of their existing blockchains. This functionality is similar to how Bitcoin wallet software allows users to see all their used Bitcoin addresses. 

\subsection{Genesis Block Authority}
\label{sec:GenesisBlockAuthority}
The Genesis Block Authority is responsible for creating genesis blocks for a ledger. The Genesis Block Authority has a similar role to a traditional certificate authority in web development. 

Before users can add data to their ledgers, they must first obtain a genesis block. The genesis block is the only block in a ledger with the previous hash field being all 0s. Uniquely, the genesis block is the only block with no previous block making it the only block that is not chained to the block before it. 

Like traditional certificate authorities, the Genesis Block Authority can verify the user's KYC (Know-your-customer) information. After verifying this data, the Genesis Block Authority can create a genesis block that identifies the user's ledger as belonging to the user. This genesis block must contain the public key of the user. This allows other services to verify the user's signature quickly. 
 
To better demonstrate the need for a Genesis Block Authority, take the example where a user utilizes our architecture to store firewall audit logs to prove to regulators that they have not been breached. The user can create a ledger to store these logs and share the ledger with the regulators. The Genesis Block Authority will create a genesis block containing the user's public key and the public key of the firewall. Thus, the regulators can quickly verify transactions' integrity on the user's ledger. 

\subsection{Executing Service}
\label{sec:ExecutingService}
The Executing Service is responsible for accepting data from the user and creating complete transactions. The data from the user may be raw data that the user wishes to store in the ledger. The data may be used as inputs to data-driven code. This is parallel to chaincode execution in models such as Hyperledger Fabric. Notice that, unlike platforms like Ethereum, which store EVM states, the raw data is stored on the ledger, making our ledger more accessible to non-technical users. 

The Executing Service transforms the users' transaction into a complete transaction by generating the output and signing the transaction. In cases where the user wishes to store raw data on their ledger without executing data-driven code, the transaction output is set to 0.  

This service is similar to the execute phase in Hyperledger Fabrics Execute-Order-Validate (EOV) model. Thus, unlike other popular blockchain development platforms like Ethereum, our model allows for flexible, extensible, nondeterministic code. Since our code does not execute on a virtual machine like the EVM, developers can create data-driven processes for our ledger using languages they already know.

\subsection{Ordering Service}
\label{sec:OrderingService}
The Ordering Service collects complete transactions, orders the transactions, and forms blocks that the validation service can eventually add to the user's ledger. 
The Ordering Service is analogous to the mempool in Bitcoin. The mempool collects transactions submitted by users, and only when a particular condition is met are some of the transactions formulated into blocks that can be added to the blockchain. Our Ordering Service accepts new complete transactions from the Executing Service and waits until a given cutting condition is met to form a block.

To decide when a new block should be created, the Ordering Service uses a cutting condition. The cutting condition defines how many transactions should be in each block. Some examples of cutting conditions are creating a block every 5 minutes, creating a block every three transactions, or creating a block when the block size is more than 1 MB. The cutting condition is modular. Thus, each service provider can choose an appropriate cutting condition for their Ordering Service.

After ordering the blocks, the Ordering Service signs the root of the Merkle Tree of the transactions. This ensures the integrity of the order set by the Ordering Service. If two transactions are reordered, the root of the Merkle Tree changes; thus, the Ordering Service's signature appears invalid. In the next section, we formally show that this immutability property holds. 
 
This order-preserving property is a fundamental feature of our ledgers. Currently, if a user wishes to store their data in cloud services, or a database, they have no guarantees that the order of the documents is preserved. Through linked cryptographic hashes, users can be assured that the order of their data is preserved.  

\subsection{Validation Service}
\label{sec:ValidationService}
Our final service is the Validation Service responsible for validating the blocks created by the Ordering Service. This service performs the validation phase of the EOV model. 

In the validation, phase blocks are checked to ensure that all the signatures presented are valid. The user's signature on each transaction, the Executing Service signature on each complete transaction, and the Ordering Service signature on the block must be valid. Likewise, the previous hash field must match the hash of the previous block's header. 

After validating the cryptographic properties of the block, the Validation Service must verify that the ordering of the blocks is valid. If the output of transaction i is used as part of the input of transaction j, transaction i must come before transaction j. Invalid ordered blocks are marked as ignored and returned to the executing service. 

After a block has been validated, the Validation Service sends the block back to the Ledger API. Through the Ledger API, the user sings the block and then sends it to the storage service to insert the block into the ledger. 

\subsection{Service Protocol}
\label{sec:ServiceProtocol}
The services described above are implemented in an independent modular fashion. That is, a user can select which service providers they trust for each service. To prevent the centralization of a user's data, users should utilize a pool of multiple trusted service providers. For example, user one may select service providers A, B, and C to provide their Executing Service. In contrast, user two may select service providers A and D to provide the same service. Users can change their trusted service providers during the life cycle of their ledgers without disrupting the ledger's availability. 

Furthermore, we assume that these are independent competing service providers. That is, service provider A does not collude with and competes with service provider B. Likewise, service provider A does not know the other service providers a user utilizes.

With the scale of modern large cloud providers, it is easy to imagine a world where companies such as Microsoft, Amazon, or Google offer these services to users, similar to how these companies already offer many cloud services to users.  

The specific provider for a service is selected at the start of each new transaction and block. After creating the genesis block, the user selects a random service provider from their trusted pool for the Validation Service. The user provides this service provider with the information necessary to form a block header, specifically the previous block's hash. This service provider is known as the Validation Service Provider (VSP).

Next, the user performs the same processes to select an Ordering Service Provider (OSP). The user will let the OSP know who the VSP is so the block can be sent to the correct entity for validation. The OSP and VSP will only change once the user has added a block.

The user will select a random service provider from their trusted pool for each transaction to act as the Executing Service Provider (ESP). The user will inform the ESP whom the OSP is so that signed transactions can be forwarded to the correct service provider. Notice that the ESP changes for each transaction; thus, when the user has multiple trusted service providers, no one service provider will have access to all of a user's data. 

After creating a block, the user will repeat the process of selecting a VSP and OSP. Again, notice that since the user randomly selects a new VSP and OSP, no service provider can access all of the user's data. Likewise, notice that the larger the pool of trusted service providers is, the less likely it is for one service provider to accumulate a user's data. 

More in-depth technical details on our service protocol, including how the services are networked, can be found on our GitHub page~\cite{CollinConnorsGitHubPage}.

In the following section, we will prove that the ledgers created through our system satisfy the five properties of blockchain. We will then discuss how our system can be applied to cover use cases missed by traditional blockchains. 

\section{Properties of Blockchains}
\label{sec:PropertiesOfBlockchains}

As discussed earlier, blockchains are defined by the five properties:
\begin{itemize}
	\item Append Only - New blocks can only be added at the end of a blockchain
	\item Immutable - Blocks and Transactions cannot be reordered
	\item Tamper Evident - If a change is made to one block, it is evident that a change was made
	\item Tamper Resistant - A blockchain resists changes by creating substantial obstacles that must be overcome to make even a small change.
	\item Fault Tolerant - If some components do not cooperate or fail the blockchain remains available for reading and writing. 
\end{itemize}

Since we have formally defined our ledgers, we can now, in this section, formally show that each of these properties holds true for our architecture. 

\subsection{Append Only} 
\label{sec:AppendOnly}

\begin{theorem}
Ledgers in our proposed system are append-only. That is, new blocks can be added at the end of the ledger (Lemma 1.1), new blocks cannot be inserted between two blocks (Lemma 1.2), and new blocks cannot be prepended to the ledger (Lemma 1.3). 
\end{theorem}

\begin{proof}
To prove this, we need to show:
\begin{itemize}
	\item Blocks can be appended to the ledger, and the ledger remains valid. (Lemma 1.1)
	\item Blocks cannot be inserted between two blocks in the ledger, and the ledger remains valid. (Lemma 1.2)
	\item Blocks cannot be prepended to the start of the ledger, and the ledger remains valid. (Lemma 1.3)
\end{itemize}
In the following lemmas, we show that each item is true. Thus, using the proposed system, ledgers are append-only.
\end{proof}

\begin{lemma}
Using the proposed system, blocks can be appended to the end of the ledger. We show this using the Append Block Algorithm shown in Algorithm 1.
\end{lemma}

\begin{algorithm}
\caption{Append a Block Algorithm}\label{alg:AppendABlockAlgorithm}
\begin{algorithmic}
	\State let $\mathcal{L}$
	\State let $b$ 
	\State $b$\texttt{.previousHash} = $h(B_n)$
	\State $b$\texttt{.validationSignature} = /
	\State \hspace{0.75in}	ValidaionService\textsc{.sign}($b$\texttt{.header})
	\State $b$\texttt{.userSignature} = User\textsc{.sign}($b$\texttt{.header})
\end{algorithmic}
\end{algorithm}

\begin{proof}
Assume a valid ledger $\mathcal{L}$ has $n$ blocks. Take an incomplete block $b$, which we wish to append to our blockchain. We will assume $b$'s data is valid, but the block is missing the previous hash, the validation services signature, and the user's signature. All other fields we assume to be valid. Notice that this is the state of new blocks after the ordering service has formed an incomplete block and sent it to the validation service. 

We calculate the cryptographic hash of block $n$'s header, hash$=h(B_n$\texttt{.header}$)$ We then assign our new block's previous hash field to be this hash $b$\texttt{.previousHash}$=$hash. Next, the validation service signs the completed block header and appends the signature to the block header. Finally, the user signs the block and appends the block to the blockchain as $B_{n+1}$.

Since $\mathcal{L}$ was a valid ledger, the genesis block is valid, and all blocks $B_1, ..., B_N$ are valid. To show that $\mathcal{L}$ is valid after appending, we must show that $B_{n+1}$ is valid and $c(B_n,B_{n+1})$ is valid.

Notice that $B_{n+1}$ is a valid block. We assumed only the previous hash, validation service signature, and the user's signature were missing, which we have now provided; thus, all fields are present. We assumed that the data hash field was already present and valid. Lastly, we assumed the Executing Service signatures and the Ordering service signatures were valid. We defined our block such that the Validation service signature is valid and the user's signature is valid; thus, all signatures are valid. 

We assume that $\mathcal{L}$ was valid, to begin with; thus, all connections between blocks $B_{i-1}$ and $B_i$ are valid for $i < n$. Thus we must show that connection $c(B_n,B_{n+1})$ is valid. We defined our block $B_{n+1}$ to validate this connection. Specifically $B_{n+1}$\texttt{.previousHash}=$h(B_n$\texttt{.header}$)$. 

Thus after modifying $b$ and appending it to our ledger as block $B_{n+1}$, $\mathcal{L}$ is a valid ledger. This shows that we can append blocks to our ledger.
\end{proof}

\begin{figure}
	\centering
		\includegraphics[width=\columnwidth]{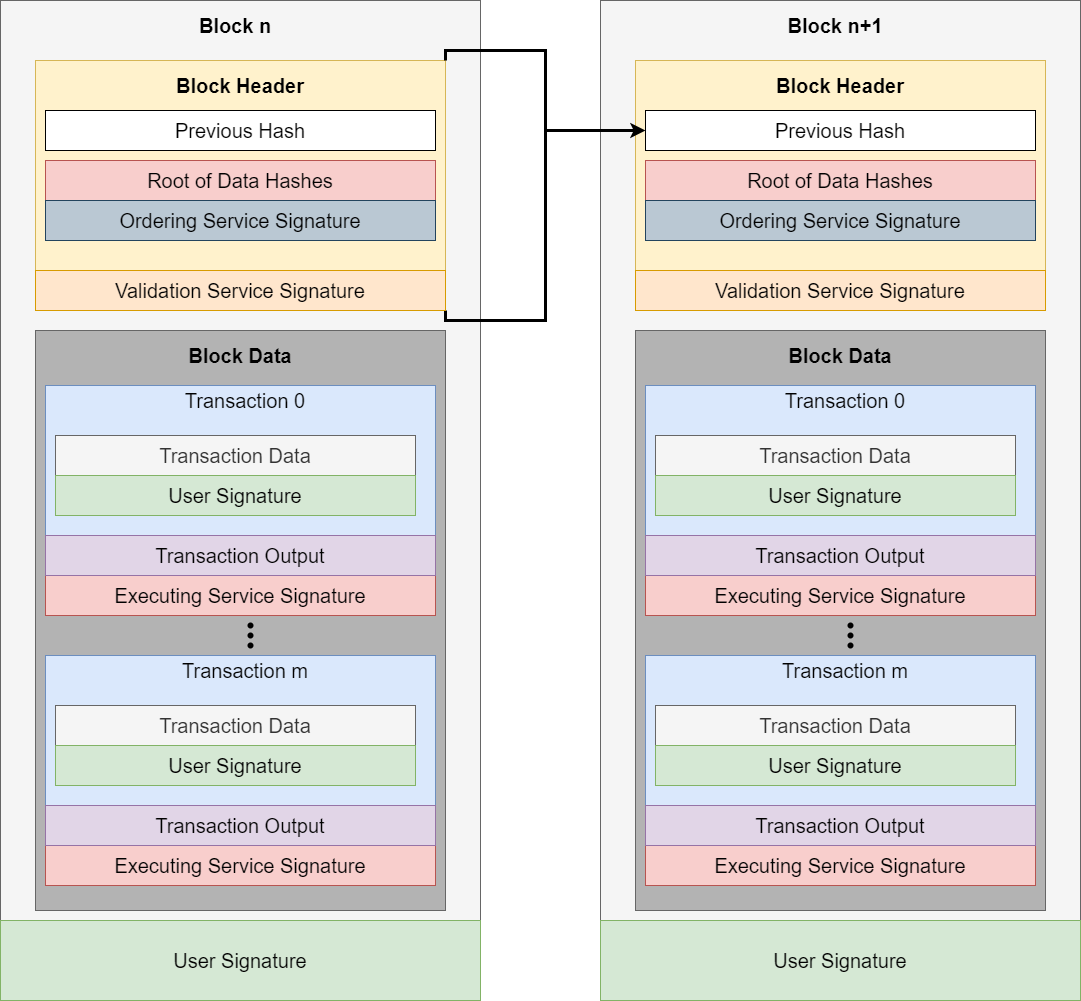}
	\caption{Visualization of the Append Property}
	\label{fig:Append}
\end{figure}

\begin{lemma}
Using the proposed system, blocks cannot be inserted between two existing blocks, and the ledger remains valid.
\end{lemma}

\begin{proof}
Assume a valid ledger $\mathcal{L}$ has $n$ blocks. Assume there is a method to insert a new block $B_j$, such that $\forall i \in [0,n] B_j \neq B_i$, into the ledger between any two blocks $B_{i}$ and $B_{i+1}$, $i < n-1$, such that the new ledger $\mathcal{L}^\prime$ is still valid after insertion. 

Since we assumed that $\mathcal{L}$ was valid before insertion, then the connection $c(B_i,B_{i+1})$ must be valid. Thus by definition $B_{i+1}$\texttt{.perviousHash} $= h(B_i$\texttt{.header}$)$

If the new ledger $\mathcal{L}^\prime$ is still valid after insertion of block $B_j$ then the connections $c(B_i,B_j)$ and $c(B_j,B_{i+1})$ must be valid. 

If connection $c(B_i,B_j)$ is valid then by definition $B_j$\texttt{.perviousHash} $= h(B_i$\texttt{.header}$)$. Similarly for $c(B_j,B_{i+1})$, $B_{i+1}$\texttt{.perviousHash} $= h(B_j$\texttt{.header}$)$. However we already showed that $B_{i+1}$\texttt{.perviousHash} $= h(B_i$\texttt{.header}$)$ thus $h(B_i$\texttt{.header}$)=h(B_j$\texttt{.header}$)$. Since we assume the use of cryptographic hashes which are unique $B_i$\texttt{.header}$=B_j$\texttt{.header}. In particular, $B_i$\textsl{.dataHash}$=B_j$\texttt{.dataHash}; since we use a Merkle Tree to generate the database field $B_i$\texttt{.data}$=B_j$\texttt{.data}. $B_i$ has the same header and data as $B_j$, $B_i=B_j$. This contradicts our assumption that $B_j$ is a new block. 

Thus new blocks cannot be inserted into our ledger.
\end{proof}

\begin{figure}
	\centering
		\includegraphics[width=\columnwidth]{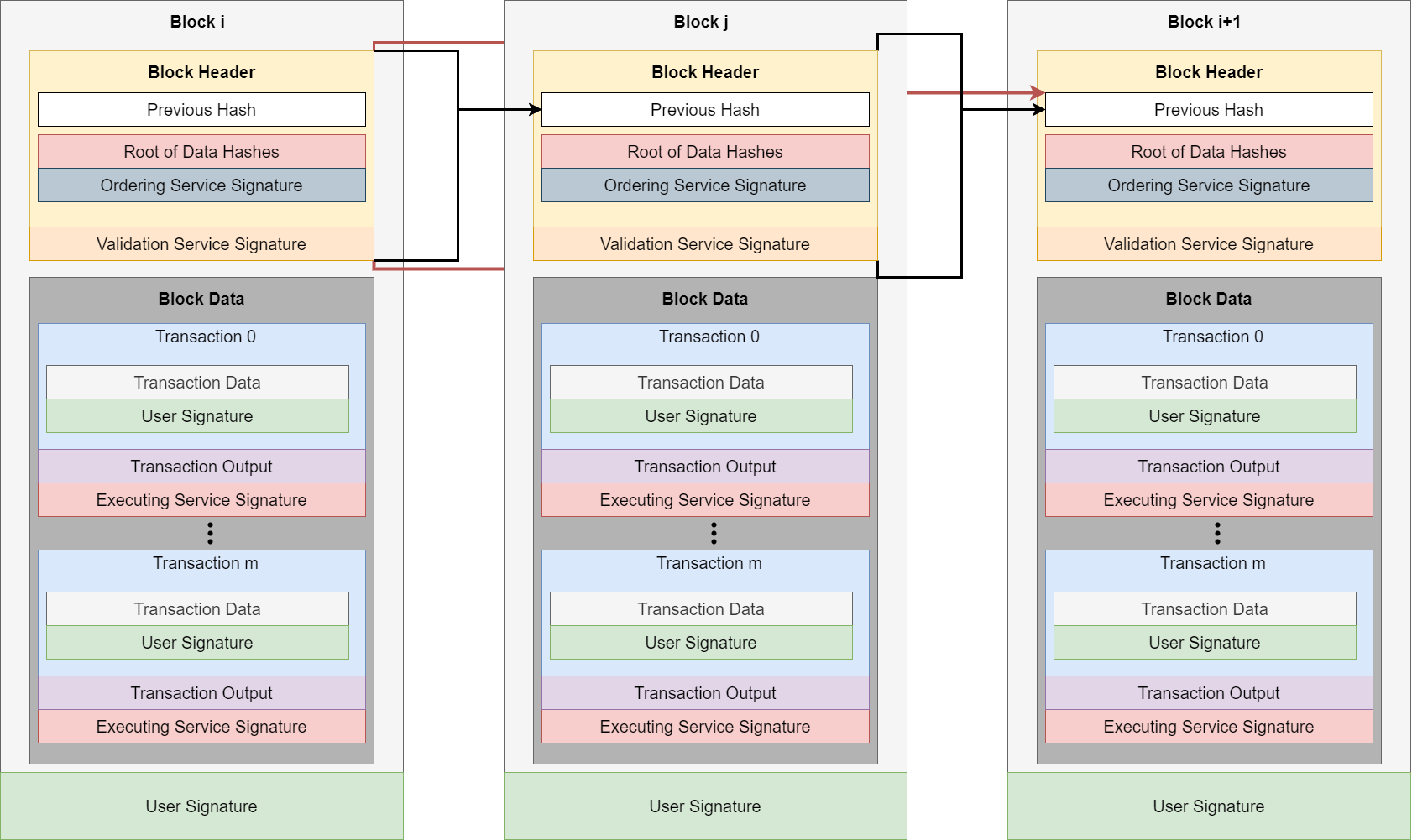}
	\caption{Visualization showing that blocks cannot be Inserted}
	\label{fig:Insert}
\end{figure}

\begin{lemma}
Using the proposed system, blocks cannot be prepended for the start of a ledger, and the ledger remains valid.
\end{lemma}

\begin{proof}
Assume a valid ledger $\mathcal{L}$ has $n$ blocks. Assume there is a method to prepended a new block $B_j$, such that $B_j \neq B_0$, into the ledger before block $B_0$, the genesis block, such that the new ledger $\mathcal{L}^\prime$ is still valid after insertion. 

Since $\mathcal{L}^\prime$ is a valid ledger, the first block must be a genesis block. That is, the block we prepended $B_j$ is a genesis block. However, $\mathcal{L}$ was valid before we prepended; thus, $B_0$ must be a genesis block. Notice that $B_0$ is still in our ledger $\mathcal{L}^\prime$. Thus $\mathcal{L}^\prime$ has two genesis blocks. This contradicts our requirement that a valid ledger has only one genesis block. 

Thus new blocks cannot be prepended to our ledger. 
\end{proof}

\begin{figure}
	\centering
		\includegraphics[width=\columnwidth]{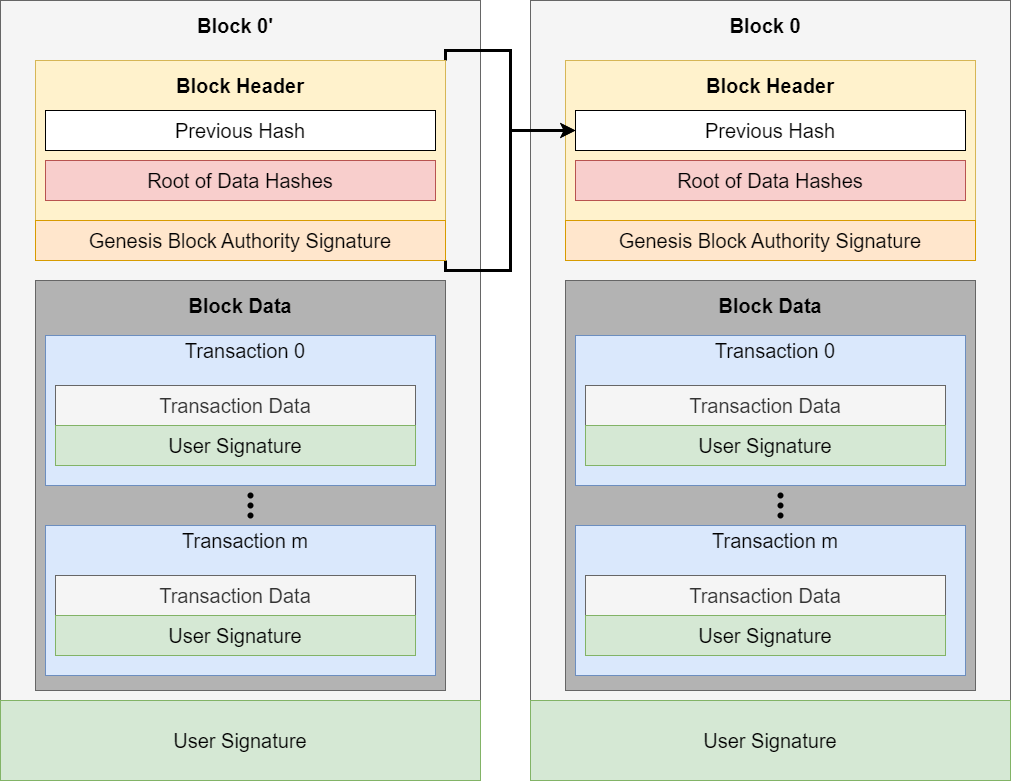}
	\caption{Visualization showing that blocks cannot be Prepended}
	\label{fig:Prepend}
\end{figure}
 
\subsection{Immutability} 
\label{sec:Immutability}

The next property of blockchain that our ledgers hold is the immutability property. To show that our ledgers are immutable, we must show that neither blocks nor transactions within a block can be reordered. If both of these are true, the data in our ledgers is immutable. 

To show this, we show that there is only one valid order of blocks; thus, blocks cannot be reordered. We then show that any reordering of the transactions causes the block header to become invalid. 

\begin{theorem}
Ledgers in our proposed system are immutable. That is, the order of blocks cannot be changed (Lemma 2.1), and the order of transactions cannot be changed (Lemma 2.2).
\end{theorem} 

\begin{proof}
To prove this, we will need to show:
\begin{itemize}
	\item Given a valid ledger, blocks cannot be reordered, and the ledger remains valid. (Lemma 2.1)
	\item Given a valid block, transactions cannot be reordered, and the block remains valid. (Lemma 2.2)
\end{itemize}
We show these two facts in the following lemmas. Thus in using the proposed system, the ledgers are immutable. 
\end{proof}

\begin{lemma}
Given a valid ledger, blocks cannot be reordered, and the ledger remains valid.
\end{lemma}

\begin{proof}
Assume a valid ledger $\mathcal{L}$ has $n$ blocks. Assume there is a method to reorder the blocks of $\mathcal{L}$ to get a new valid ledger $\mathcal{L}^\prime$ such that $\mathcal{L} \neq \mathcal{L}^\prime$.

Since $\mathcal{L}^\prime$ is a valid ledger, the first block must be a genesis block. Since there is only one genesis block in $\mathcal{L}$, $B_0$, then the genesis block of $\mathcal{L}^\prime$ must also be $B_0$. 

Assume that block two blocks are the same in $\mathcal{L}$ and $\mathcal{L}^\prime$, that is $B_i = B^\prime_i$. We want to show that the next blocks must also be equal $B_{i+1} = B^\prime_{i+1}$ since $\mathcal{L}$ is a valid ledger the connection $c(B_i,B_{i+1})$ is valid. Likewise for $\mathcal{L}^\prime$ $c(B^\prime_i,B^\prime_{i+1})$ is valid. Thus $B_{i+1}$\texttt{.perviousHash} $= h(B_i$\texttt{.header}$) = h(B^\prime_i$\texttt{.header}$) = B^\prime_{i+1}$\texttt{.previousHash}. 

Since $h()$ is a cryptographic hash function, the hashes produced are unique. That means only one block in $\mathcal{L}$ can have a given hash in its previousHash field. Since $\mathcal{L}^\prime$ is a reordering of $\mathcal{L}$ this means that $B_{i+1} = B^\prime_{i+1}$.

Thus by mathematical induction $\mathcal{L} = \mathcal{L}^\prime$. This contradicts our assumption that $\mathcal{L} \neq \mathcal{L}^\prime$. Thus blocks cannot be reordered.
\end{proof}

\begin{figure}
	\centering
		\includegraphics[width=\columnwidth]{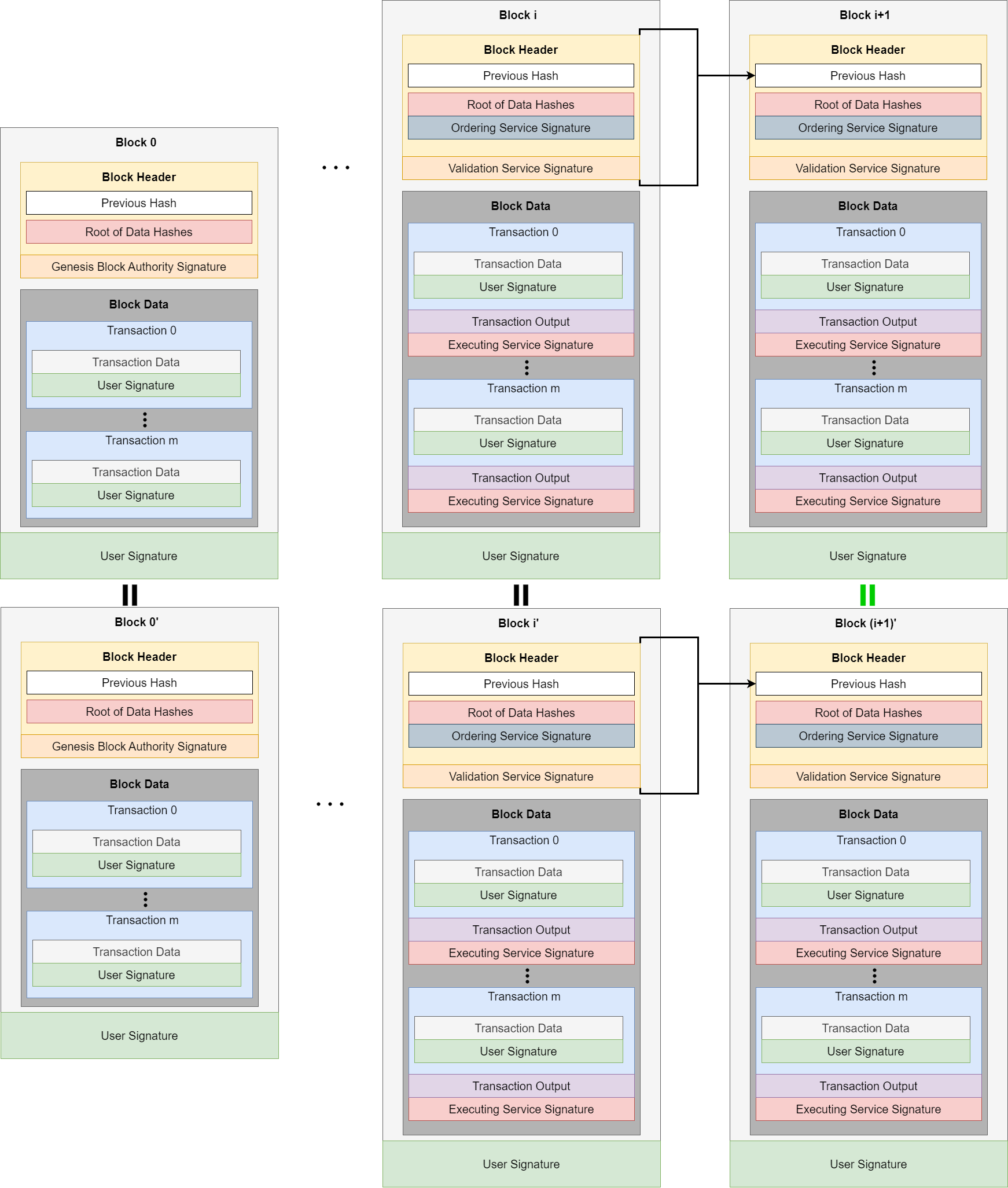}
	\caption{Visualization showing that blocks cannot be reordered}
	\label{fig:BlockOrder}
\end{figure}

\begin{lemma}
Given a valid block, transactions cannot be reordered, and the block remains valid.
\end{lemma}

\begin{proof}
Assume a valid ledger $\mathcal{L}$ has $n$ blocks. Assume there is a method to reorder the transactions $[t_0,...,t_m]$ of any arbitrary block $B_i$ to get a new valid block $B^\prime_i$. 

Notice all that was changed was the order of transactions, thus $B_i$\texttt{.header}$=B^\prime_i$\texttt{.header}. In particular $B_i$\texttt{.dataHash}$=B^\prime_i$\texttt{.dataHash}. Remember that the dataHash is generated from the Merkle Tree of transactions. 

Based on the properties of Merkle Trees and cryptographic hashes, the root is only the same if all of the inputs are in the same order. Thus the Merkel Tree generated by block $B_{i}$ transactions cannot be the same as the one generated from the reordered transactions in $B^\prime_{i}$. This contradicts the fact that they must have the same dataHash field.

Thus transactions cannot be reordered within a block.
\end{proof}

\begin{figure}
	\centering
		\includegraphics[width=\columnwidth]{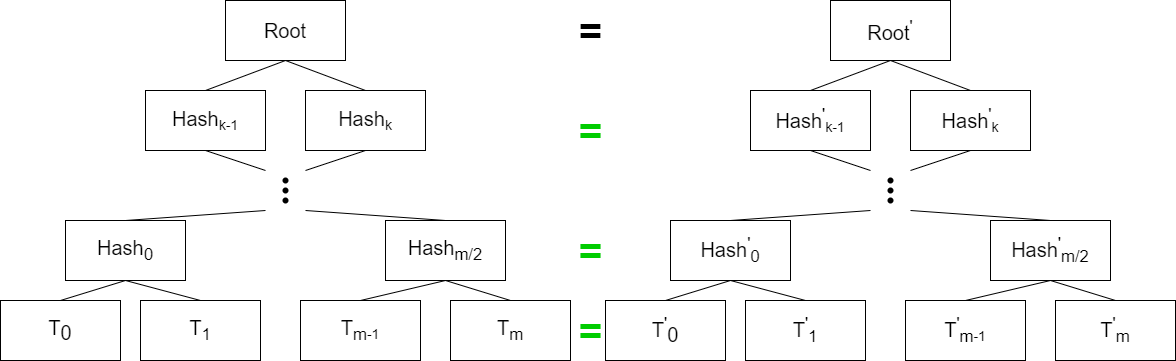}
	\caption{Visualization showing that transactions cannot be reordered}
	\label{fig:BlockOrder}
\end{figure}

\subsection{Tamper Evident}
\label{sec:TamperEvident}
Next, we show that our ledgers are tamper evident. That is, if any data is changed within one block, it is evident that a change was made on the rest of the ledger. In this lemma, we show this by first showing that any change to the block causes a change in the block's header. Next, we show that the change in the header is evident in the next block for all but the last block. Finally, we show that this change is evident in the user's signature in the last block. 

For our tamper-evident requirement, it is important to remember that our ledgers are for personal use. That is, the user is the only stakeholder in the ledger. Thus we assume that the user will refrain from tampering with their own ledger.

\begin{theorem}
Ledgers in our proposed system are tamper-evident. Any change to one block will be evident in subsequent blocks. 
\end{theorem} 

\begin{proof}
Assume a valid ledger $\mathcal{L}$ has $n$ blocks. Assume there is a method to modify data within a block $B_i$ such that $i \in [0,n)$ that produces a valid block $B^\prime_i$ on the new valid ledger $\mathcal{L}^\prime$.

First, notice that any changes to the block must change the block header. Clearly, any change to the block header itself changes the block header. Notice that any change to the block data will also require changing the hash of the block data stored in the header. If the hash of the block data were not changed, it would be evident that the block data was changed, contradicting the assumption that it is not tamper evident. 

Since any change to the block changes the header, then our $h(B^\prime_i) \neq h(B_i)$. However since $\mathcal{L}$ and $\mathcal{L}^\prime$ are valid ledger the connections $c(B_i,B_{i+1}$ and $c(B^\prime_i,B_{i+1}$ are valid. By definition $B_{i+1}$\texttt{.previousHash} $= h(B_i)$ and $B_{i+1}$\texttt{.previousHash} $= h(B^\prime_i)$. This implies $h(B^\prime_i) = h(B_i)$. However, we have already shown that this cannot be true.

Next, we must show that $B_n$ cannot be modified. Notice that our previous logic required the modified block to have a subsequent block; thus, it does not apply to $B_n$. However, the user signs all blocks in our system before being added to the ledger. Thus, if anyone other than the user modifies the block, they cannot modify the user's signature such that the signature is still valid.

Thus any changes to a block are evident on our ledgers.
\end{proof}

\begin{figure}
	\centering
		\includegraphics[width=\columnwidth]{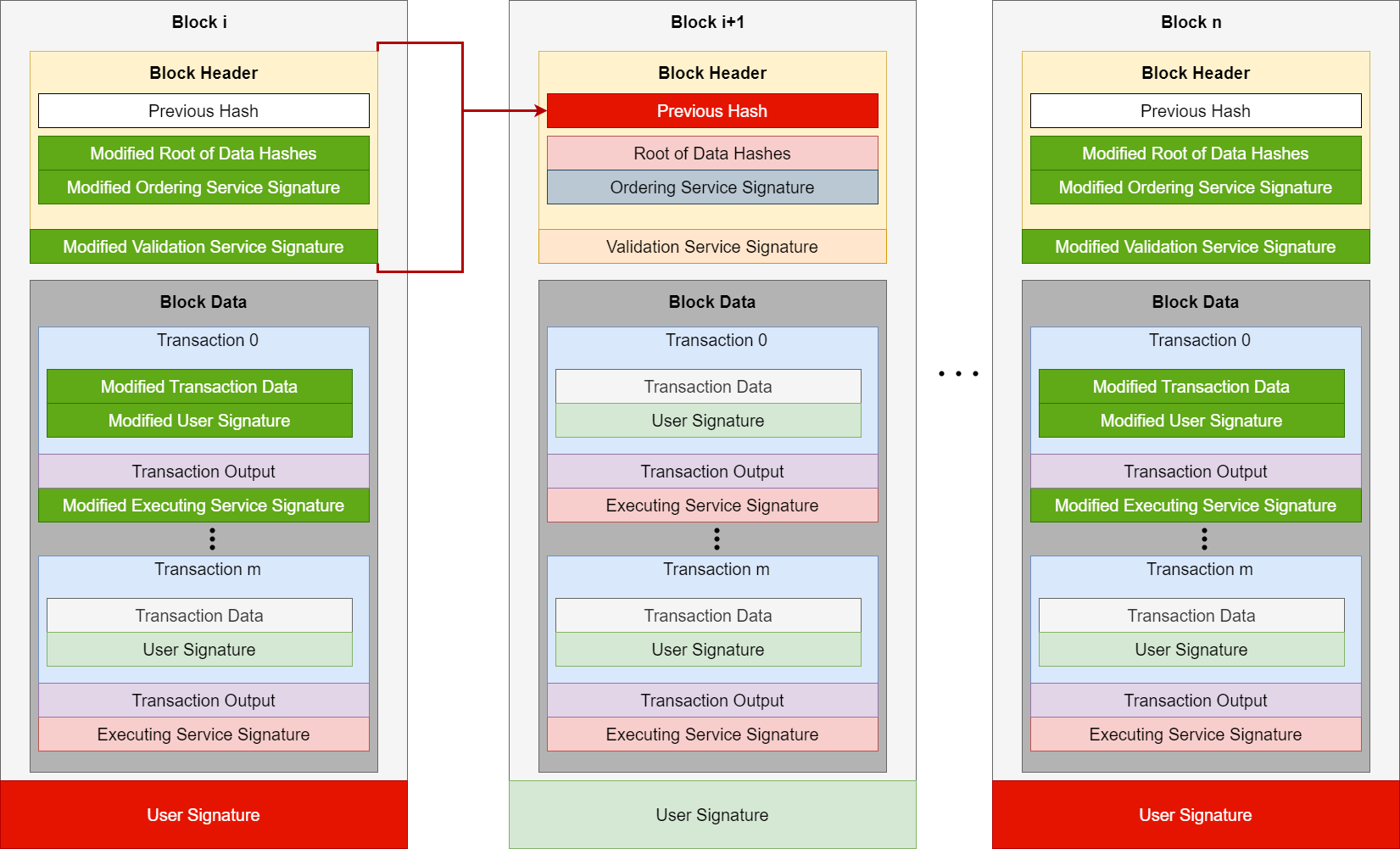}
	\caption{Visualization showing that changing the block data in one block cause the entire ledger to be invalid.}
	\label{fig:TamperEvident}
\end{figure}

\subsection{Tamper Resistant}
\label{sec:TamperResistant}
We have already shown that any changes to an individual block cause the block's user signature field to be invalid and the next block's previous hash field to become invalid. However our ledgers resist this type of tampering in two ways. First any changes to a blocks data requires all of our services to collude in order for the change to be made. This property comes from the fact that each transaction is signed by each service thus any changes will require all of the services to resign the change. Likewise, if all of the services were to resign a change in one block they would be required to modify all subsequent blocks wich will take O(n) time.

\begin{theorem}
In our proposed system ledgers resist tampering by requiring all services to agree on changes (Lemma 4.1) and if all services do agree on a change it would require O(n) time to modify the ledger (Lemma 4.2).
\end{theorem}

\begin{proof}
To prove this we will need to show:
\begin{itemize}
	\item Given a valid ledger, any changes to the data will require all services to agree on the change. (Lemma 4.1)
	\item Given a valid ledger, if a change is made to one block it will require O(n) time to modify the entire ledger. (Lemma 4.2)
\end{itemize}
We show these two facts in the following lemmas. Thus in using the proposed system the ledgers are tamper-resistant. 
\end{proof}

\begin{lemma}
In the proposed system any changes to the data will require all services to agree on the change in order for the block to remain valid.
\end{lemma}

\begin{proof}
Assume a valid ledger $\mathcal{L}$ has $n$ blocks. Assume there is a method to modify data within a block that does not require all of the services to collude $B_i$ such that $i \in [0,n]$ that produces a valid block $B^\prime_i$ on the new valid ledger $\mathcal{L}^\prime$. 

Notice that each transaction is signed by the Executing Service. For $\mathcal{L}^\prime$ to remain valid all of the Executing Service signatures must remain valid.  Since the transactions are signed with a private key only the Executing Service can resign the transaction. Thus if transaction data is changed the Executing Service must resign the changed transaction.

Likewise, any changes to the transactions will change the root of the Merkle Tree. Since the Ordering Service signs the root of the Merkle Tree the Ordering Service must also agree to resign the block. If the Ordering Service does not collude the signature will be invalid, thus the block will be invalid. 

This same logic applies to the Validation Service. Since the Validation Service must sign the block header then the Validation Service is required to collude.

Lastly all blocks are signed by the user before they are added to the ledger. Thus by the same logic the user must also collude. 

We have shown that the Executing Service, the Ordering Service, the Validation Service, and the User must all collude to create a new valid block. Thus all services must be willing to collude in order to modify a single block. This shows our ledger resist tampering by requiring multiple independent services to conspire to modify blocks. 
\end{proof}

While it is unlikely that all of the services would conspire to modify a block it is possible; thus our system is designed to further resist tampering by requiring substantial work to modify a single block. 

\begin{lemma}
In the proposed system if a change is made to one block it will require O(n) time to modify the entire ledger to overcome the tamper-evident property.
\end{lemma}

\begin{proof}
To show that an $f(n)=O(g(n))$ we must show that $\exists n_0$ such that $0 \leq f(n) \leq c*g(n)$ $\forall n>n_0$. \cite{BigO}

Let $f(n)$ be the function representing the time it takes to modify the ledger to overcome the tamper-evident property. It takes some constant time $c_0$ to modify block $B_i$. Assuming that all parties are willing to resign the block. 

Then for each subsequent block $B_j$, it will take $k_1+k_2+k_3$ time to modify the block where $k_1$ is the time to modify the previous hash field such that it is valid, $k_2$ is the time to modify the validation service signature such that it is valid, and $k_3$ is the time it takes to modify the user's signature such that it is valid. We assume that whoever is tampering with the ledger has access to any necessary private keys. Notice $k_1,k_2,k_3$ are all constant time; thus, $k_1+k_2+k_3$ can be done in constant time $c_1$. Again this assumes that the validation service and the user are willing to resign the block. 

All blocks from block $B_{i+1}$ to block $B_{n}$ need to be modified this way. Thus $f(n) = c_0 + (c_1) * (n-i+1)$. We get that $f(n)=c_1*n-c_1*i+c_1+c_0$ Notice that $-c_1*i+c_1+c_0$ is constant we will call this constant $c_2$. Notice $i<n$ and $c_0>0$ implies $c_2>0$. $f(n) = c_1*n+c_2$

Let $g(n) = n$. Let $c$>$c_1$. Let $n_0 = 0$. 

$0 \leq f(n) = c_1*n+c_2 \leq c1*n \leq c*n = g(n)$. 

Thus $f(n) = O(g(n)) = O(n)$.
\end{proof}

\begin{figure}
	\centering
		\includegraphics[width=\columnwidth]{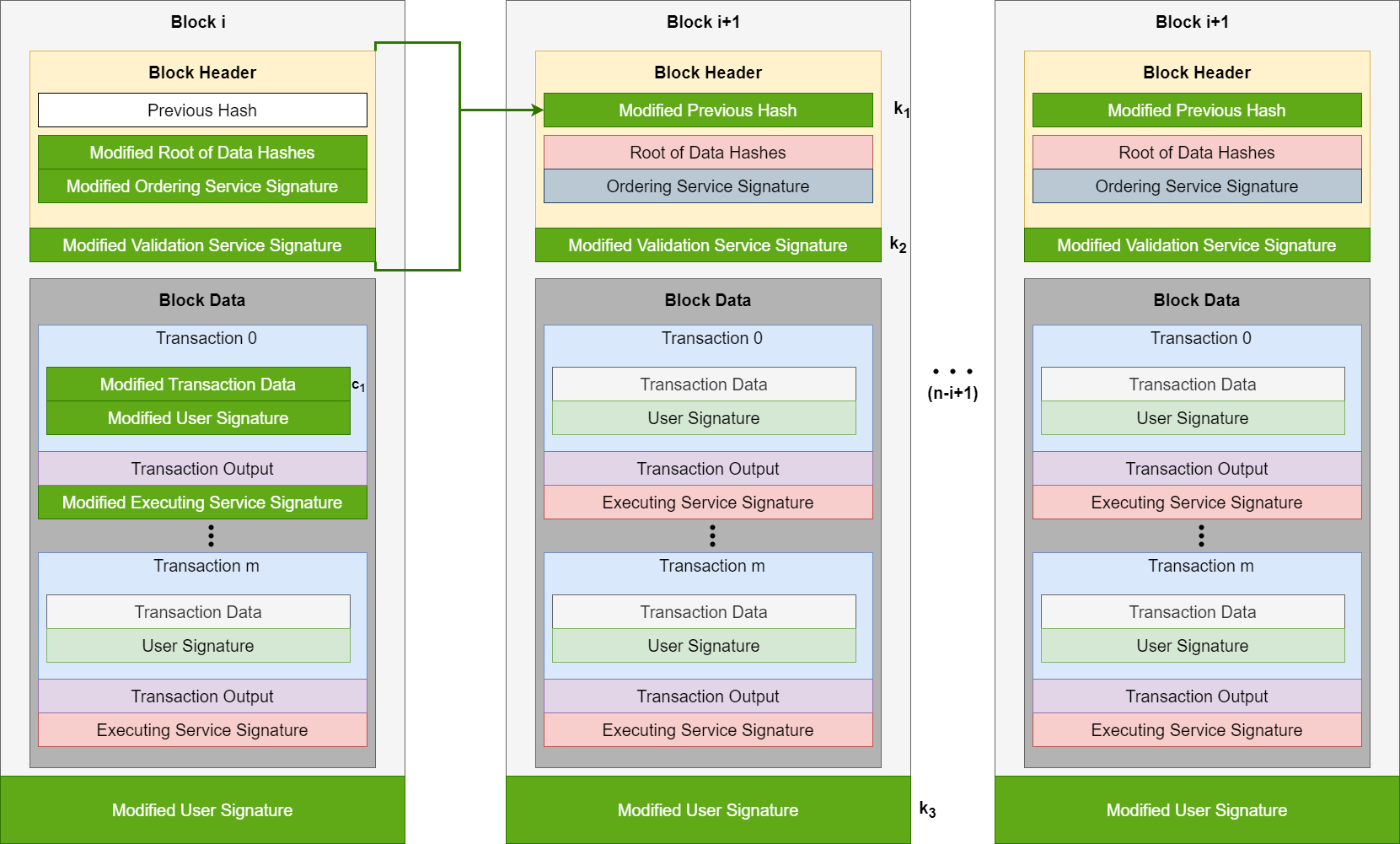}
	\caption{Visualization showing that it will take O(n) time to modify our ledgers.}
	\label{fig:TamperResistant}
\end{figure}

\subsection{Fault Tolerant} 
\label{sec:FaultTolerent}
The final property of blockchain that our ledgers satisfy is the fault-tolerant property. A fault occurs when one or more components of a system fail to work. In our system, a component is said to fail to work if the component does not respond to a request within a given Time To Live (TTL). A fault could be due to the component dying, having a slow network, refusing to participate, or another reason. To be fault tolerant in the context of our system is for a ledger to be available for reading and writing after a fault has occurred. 

Our system's components are not the full service but individual service providers. This is because the users utilize multiple independent service providers to provide services. For example, users may select service providers A and B to provide their Executing Service. A fault occurs if service provider A fails to respond to a request within a given TTL. Because service providers A and B are independent, we assume there is no correlation between the probability of service provider A faults and the probability of service provider B's faults.

\begin{theorem}
Our system is fault tolerant. If a component faults, the ledger is still available to read (Lemma 5.1), and the ledger is still available to write (Lemma 5.2). 
\end{theorem}

\begin{proof}
To prove this we will need to show:
\begin{itemize}
	\item Given a valid ledger, after a fault occurs the ledger is still available to read. (Lemma 5.1)
	\item Given a valid ledger, after a fault occurs the ledger is still available to write. (Lemma 5.2)
\end{itemize}
We show these two facts in the following lemmas. Thus in using the proposed system the ledgers are fault tolerant. 
\end{proof}

\begin{lemma}
In our proposed system, the ledger is still available to read if a fault occurs.
\end{lemma}

To prove this lemma, we assume that the user is storing their ledger in a trusted manner. Remember that the storage service does not store the ledgers but acts as an interface between the user and their desired storage mechanism. As stated, users may choose how their ledgers are stored, such as in IPFS, on existing cloud storage, or even on their own devices. The user may opt to use multiple storage locations for redundancy. This gives the users more control over their data. We assume that the user's selected storage location is always available. We discuss this assumption more in section~\ref{sec:Limitations}.

\begin{proof}
Assume a user has a list of acceptable service providers for each service. Assume that these lists contain $m$ service providers such that $m>1$. Lastly, we assume the user has selected a trusted method to store their ledger.

A fault can occur in any of our proposed services. The ledger can still be read if up to $m$ faults occur in any of the Executing  Service, Ordering Service, Validation Service, or Genesis Block Authority. Since these services are responsible for cryptographically signing blocks in the ledger, we need to show that even when $m$ faults occur the in these services, their signatures can still be verified. However, since these services sign using their private key to verify their signature, we need access to their public key. Remember that the Ledger API adds each service provider's public key to the Root Address; thus, even if the service does not respond, the user can still access the public key to verify the service's signature.  

Furthermore, the user can still access their ledger if up to $m-1$ faults occur in the Storage Service. If $m-1$ faults have occurred, then exactly 1 Storage Service provider is available. Since the provider has not faulted, the user can query this Storage Provider to access their ledger normally. Thus the user can access their ledger.

This shows that the ledger is available to read as long as no more than $m-1$ faults occur. 
\end{proof}

Notice that this proof shows that as long as at least one service provider for the storage service is available, the user will be able to read their ledger. Thus at maximum, all of the service providers for the Executing Service, Ordering Service, Validation Service, and Genesis Block Authority can fault. All but one of the service providers for the Storage Service can fault, and the ledgers are still available to read. Notice that the number of allowed faults is proportional to the number of service providers a user utilizes. Thus to ensure the availability of their ledger, a user should use a multitude of service providers. 

\begin{lemma}
In our proposed system, the ledger is still available to write if a fault occurs.
\end{lemma}

\begin{proof}
Assume a user has a list of acceptable service providers for each service. Assume that these lists contain at least $m>1$ service providers. Lastly, we assume the user has selected a trusted method to store their ledger.

A fault can occur in any of our proposed services. The ledger can still be written to if up to $m-1$ faults occur in any of the Executing  Service, Ordering Service, Validation Service, or Genesis Block Authority. Since at least one service provider has not faulted, the user can select a valid service provider at the start of each round. Thus the user can always send their request to a valid service provider. Thus these services can continue to add new blocks to the ledger. 

Furthermore, the user can still access their ledger if up to $m-1$ faults occur in the Storage Service. If $m-1$ faults have occurred, then exactly 1 Storage Service provider is available. Since the provider has not faulted, users can send post requests to this Storage Provider to access their ledger normally. Thus the user can add new blocks to their ledger.

This shows that the ledger is available to write as long as no more than $m-1$ faults occur. 
\end{proof}

Notice that the system's fault tolerance is a function of $m$. Thus the more service providers a user chooses, the more faults our system can tolerate. 

Using the previous five theorems, we have proved that ledgers generated by our proposed systems maintain all blockchains' properties. We proved these ledgers are append-only, immutable, tamper-evident, tamper-resistant, and fault tolerant. Thus we can state that our system creates personal blockchain ledgers. The following section will discuss why such a system is needed to solve use cases where traditional blockchains fail. 

\section{Discussion}
\label{sec:Discussion}

\subsection{Use Cases} 
\label{sec:UseCases}
In previous sections, we presented our proposed modular system, which relies on independent services to achieve the desirable properties of blockchain. We have proved that our system enables the creation of append-only, immutable, tamper-evident, tamper-resistant, and fault-tolerant ledgers. Since we have shown that this system creates valid blockchain ledgers, it begs the question of when the system would be desirable over traditional blockchain systems. To highlight this, we present an example application for storing user financial transactions. 

Financial systems have become more accessible and specialized. This specialization has led consumers to utilize multiple different financial intuitions. The complex web of various financial institutions has made it difficult for consumers to keep track of all their financial transactions. For example, an individual may have multiple banks, credit cards, and investments. Figure~\ref{fig:Traditional} shows how a customer interacts with existing financial institutions. The customer is receiving data from many different financial institutions. The consumer is responsible for collecting all of these sources of financial data, which can be time-consuming, tedious, and prone to mistakes. 

\begin{figure}
	\centering
		\includegraphics[width=.45\textwidth]{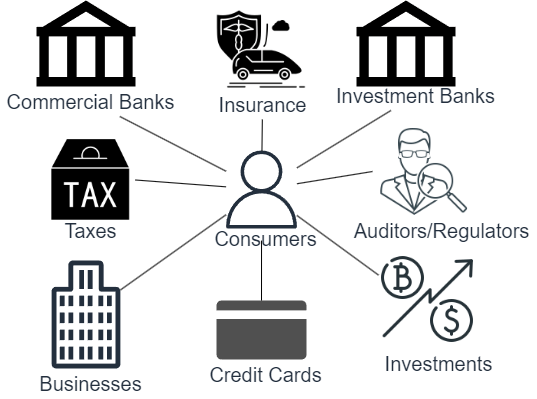}
	\caption{A high-level view of the traditional financial institutions. In this model, customers collect and maintain their financial transactions.}
	\label{fig:Traditional}
\end{figure}

We propose an application that aggregates all their financial transactions across multiple financial institutions to aid consumers. Our proposed application relies on existing APIs provided by financial institutions to collect transaction data. The application then stores the financial transactions in a blockchain. The blockchain makes users' financial records immutable and forces transactions to be chronological, two highly desired properties when tracking financial records. Figure~\ref{fig:HighLevel} gives a high-level overview of such an application. 

\begin{figure}
	\centering
		\includegraphics[width=.45\textwidth]{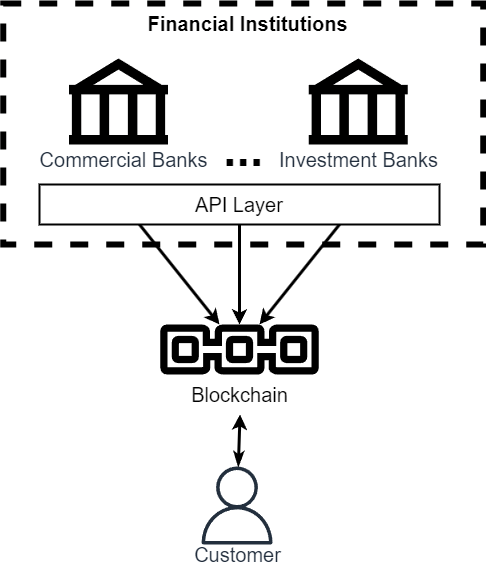}
	\caption{A high-level overview of the proposed blockchain-augmented financial application. This model stores a customer's financial transactions on a blockchain. We assume the financial institutions already have an existing API layer to send transaction data.}
	\label{fig:HighLevel}
\end{figure}

Some non-blockchain-based applications, such as Mint by Intuit~\cite{intuitMintApplictation}, have already attempted to solve this problem. These applications aggregate a user's financial transitions to the application. While this does make tracing complex financial transactions easier for consumers, these applications provide convenience to the customer at the cost of privacy. Users must hand over their finical records to a centralized third party. This third party may use the financial transaction data for advertising or other purposes. Even if the third party does not utilize the data, if a hacker were to breach the third party, the customer's sensitive financial data could get leaked. Using a centralized application requires consumers to take on risk, which is undesirable in financial applications. 

Likewise, a similar application that allows users to store data on their devices or cloud infrastructure does not provide all the desirable features. Critically past financial history should never change. Thus any storage should be immutable. Likewise, financial transaction history must be ordered. For example, consumers need to know if they have funds in their accounts before purchasing. Thus the consumer must ensure that they have funds deposited before they make a purchase. In addition, no party should have access to a consumer's financial records other than the consumer and the financial institutions. Considering these factors, blockchain presents itself as an ideal solution. However, current blockchain implementations are suboptimal for this proposed application.

If our proposed application were to use a permissionless blockchain such as Ethereum, users would have to publish their financial transactions to the blockchain. Since anyone can read a public blockchain, this could expose sensitive financial transactions. Even if the user were to take privacy measures such as encrypting the transactions, they would still risk leaking sensitive information. User error, mismanaged keys, or new technologies can reveal encrypted data. Lastly, users must pay a fee for every transaction in a permissionless blockchain. Even on blockchains with low fees, the cost can dissuade individuals and businesses alike from using permissionless blockchains. 

In contrast, if our proposed application used an existing permissioned blockchain, such as Hyperledger Fabric, the user would be required to create and maintain their own blockchain infrastructure. Creating and maintaining infrastructure can be costly and complex. Likewise, users are responsible for managing permissions to the blockchain. If a user were to make a mistake with the permission set, they could grant access to parties who should not have access. A permissioned blockchain requires technical skill and computing power and is prone to mistakes making it a suboptimal solution for an application aimed at average consumers.

Blockchain database solutions such as Amazon	QLDB are suboptimal for this use case. QLDB requires the user to give all their information to a single service provider, in this case, Amazon. This locks the user into a service provider requiring the user to provide Amazon with all of their data. An ideal solution would give users full control over who has access to their data and allow users to switch service providers without disrupting the blockchain.

Our system solves these problems presented by traditional blockchains. Like permissionless blockchains, our system is highly accessible through independent modular services. Users in our system do not need extensive technical knowledge to create and maintain their blockchains. However, like permissioned blockchains, the user's private data can only be accessed by trusted parties. Likewise, our modular system prevents vendor lock-in and allows for increased data privacy. By blending the desirable properties of permissioned and permissionless blockchains, our system allows for more unique use cases that traditional models do not address. 

Our architecture provides a new paradigm, allowing blockchain technology to expand to more use cases. Just as private blockchains expanded the possible use cases for public blockchains, we envision our individual blockchains will further diversify blockchain's use cases. 

\subsection{Limitations} 
\label{sec:Limitations}

While our system covers many use cases not addressed by traditional blockchain development platforms, it still faces some limitations, including not being suited for all blockchain applications and passing some responsibility to the users. 

Our system provides users with a high amount of control over their data. While this is a desirable feature for preserving privacy, it can be dangerous. Like with all blockchains, the data stored on our ledgers can not be changed after it has been added. To ensure the user approves all data added, we require the user to sign all transactions and blocks. However, a careless user may still sign something they did not intend. In our system, the user is responsible for ensuring the correctness of the data they are adding to their ledgers. 

Similarly, our system allows users to choose their storage system, which, if chosen poorly, can lead to a loss of availability. The ability to choose the storage system makes our system highly flexible, which is desirable for many users. However, this flexibility allows users to make mistakes that could lead to complications with their ledger. The Storage Service does provide features, such as interfacing with multiple storage locations and allowing users to change storage locations, but a careless user may ignore these features. While there are some safeguards to help users, it is up to the user to ensure their ledgers are configured correctly. 

A common application of blockchain technology is cryptocurrencies such as ERC20 Tokens. Since our ledgers are designed for personal use, creating a cryptocurrency using our system is infeasible. While cryptocurrency has rapidly increased the popularity of blockchain, we believe future blockchain applications will not resemble these early use cases. Blockchain technology can be expanded to a more diverse set of applications. Thus, we designed our system to address use cases not addressed by traditional blockchains. 

It is critical to note that we do not propose replacing existing blockchain development platforms. Platforms like Ethereum and HyperLedger Fabric cover many use cases for our system is not designed for. Rather, this work aims to propose a new blockchain system that supplements current blockchain systems, allowing for more diverse blockchain applications. 

\section{Conclusion} 
\label{sec:Conclusion}
This work defined a novel system for creating personal blockchain ledgers. Our system relies on six novel independent modular services to provide users with personal blockchains. Critically, we then proved that our ledgers maintain the five properties of the blockchain. Finally, we highlighted how our system provides a new paradigm for creating blockchain applications. 

It is important to note that our ledgers are designed for personal use. That is, the only stakeholder in the ledger is the user. Thus, unlike traditional public blockchains, users can store sensitive data on their ledger. Similarly, unlike traditional private blockchains, users do not need to set up their own blockchain system to use our ledgers. Our system makes it simple and secure for non-technical users to utilize blockchain technology. 

In future work, we plan to show how these personal ledgers can be shared securely. This will allow third-party entities to view our ledger data while the user continues to maintain control of this data. Likewise, we plan to release a complete implantation of this system on our GitHub page\cite{CollinConnorsGitHubPage}.

Overall our system will help make blockchain technology more assessable to users by allowing for more diverse blockchain applications. We hope our system provides a new paradigm allowing even non-technical users to use blockchain technology to safely store and maintain their documents.

\bibliographystyle{plain} 
\bibliography{refs}
\end{document}